\documentclass[11pt]{article}

	\usepackage{a4,geometry}
\usepackage{graphicx}
\usepackage{amsmath,amssymb,amsthm,mathtools}
\usepackage{paralist}
\usepackage{bm}
\usepackage{xspace}
\usepackage{url}
\usepackage{fullpage}
\usepackage{boxedminipage}
\usepackage{wrapfig}
\usepackage{ifthen}
\usepackage{color}
\usepackage{xcolor}
\usepackage{framed}
\usepackage{algorithmic,algorithm}
\usepackage{fullpage}

\usepackage{thmtools}
\usepackage{thm-restate}
\usepackage[pagebackref,letterpaper=true,colorlinks=true,pdfpagemode=none,urlcolor=blue,linkcolor=blue,citecolor=violet,pdfstartview=FitH]{hyperref}
\newtheorem{theorem}{Theorem}[section]

\newtheorem{claim}[theorem]{Claim}

\newtheorem{definition}[theorem]{Definition}

\newcommand{\ignore}[1]{}

\newcommand{\R}{\mathbb R}

\newcommand{\eps}{\varepsilon}

\newcommand{\bb}{\boldsymbol{b}}

\newcommand{\cei}[1]{\lceil#1\rceil}

\newcommand{\Sec}[1]{\hyperref[sec:#1]{\S\ref*{sec:#1}}} %section
\newcommand{\Eqn}[1]{\hyperref[eq:#1]{(\ref*{eq:#1})}} %equation
\newcommand{\Fig}[1]{\hyperref[fig:#1]{Fig.\,\ref*{fig:#1}}} %figure
\newcommand{\Tab}[1]{\hyperref[tab:#1]{Tab.\,\ref*{tab:#1}}} %table
\newcommand{\Thm}[1]{\hyperref[thm:#1]{Theorem\,\ref*{thm:#1}}} %theorem
\newcommand{\Fact}[1]{\hyperref[fact:#1]{Fact\,\ref*{fact:#1}}} %fact
\newcommand{\Lem}[1]{\hyperref[lem:#1]{Lemma\,\ref*{lem:#1}}} %lemma
\newcommand{\Prop}[1]{\hyperref[prop:#1]{Prop.~\ref*{prop:#1}}} %property
\newcommand{\Cor}[1]{\hyperref[cor:#1]{Corollary~\ref*{cor:#1}}} %corollary
\newcommand{\Conj}[1]{\hyperref[conj:#1]{Conjecture~\ref*{conj:#1}}} %conjecture
\newcommand{\Def}[1]{\hyperref[def:#1]{Definition~\ref*{def:#1}}} %definition
\newcommand{\Alg}[1]{\hyperref[alg:#1]{Alg.~\ref*{alg:#1}}} %algorithm
\newcommand{\Ex}[1]{\hyperref[ex:#1]{Ex.~\ref*{ex:#1}}} %example
\newcommand{\Clm}[1]{\hyperref[clm:#1]{Claim~\ref*{clm:#1}}} %example
\newcommand{\Step}[1]{\hyperref[step:#1]{Step~\ref*{step:#1}}} %example

\begin{document}
\title{A {$\widetilde{O}(n)$} Non-Adaptive Tester for Unateness }
\date{}

\author{Deeparnab  Chakrabarty \\
 Microsoft Research Bangalore\\
{\tt dechakr@microsoft.com}
\and
C. Seshadhri \\
University of California, Santa Cruz\\
{\tt sesh@ucsc.edu}
}

\def\hf{\hat{f}}
\def\hg{\hat{g}}
\def\Inf{{\sf Inf}}
\def\Viol{{\sf Viol}}
\def\I{{\mathbf I}}
\def\V{{\mathbf V}}
\def\R{{\mathbf R}}
\def\bb{\mathbf{b}}

\newcommand{\unatetest}{{\tt Unate-test}}

\maketitle
\begin{abstract} Khot and Shinkar (RANDOM, 2016) recently describe an adaptive, $O(n\log(n)/\varepsilon)$-query tester for unateness of Boolean functions $f:\{0,1\}^n \mapsto \{0,1\}$. In this note, we describe a simple non-adaptive, $O(n\log(n/\varepsilon)/\varepsilon)$ -query tester for unateness for real-valued functions over the hypercube.
\end{abstract}

\def\p{\mathbf p}
\def\B{\{0,1\}^n}

\section{Introduction} \label{sec:intro}

Let $f$ be a function $f:\{0,1\}^n \mapsto R$ defined over the Boolean hypercube where $R$ is some ordered range.
We use $e_i$ to denote the unit vectors in $\{0,1\}^n$ that has
$1$ in the $i$th coordinate, and $0$s at other coordinates. The \emph{$i$-th
partial derivate} at $x$ is $f(x\oplus e_i) - f(x)$ if $x_i = 0$
and $f(x) - f(x\oplus e_i)$ is $x_i = 1$. This is denoted by the
function $\partial_i f$. Note that a function is monotonically increasing (or simply monotone),
if $\partial_i f(x) \geq 0$ for all $x\in \{0,1\}^n$. 

Unateness is a generalization of monotonicity. A function is unate if in every coordinate it is either monotone or anti-monotone.
More precisely, a function $f$ is unate if for all $i \in [n]$, either
$\partial_i f(x) \geq 0$ for all $x$, or $\partial_i f(x) \leq 0$ for all $x$. The problem of unateness testing was introduced by Goldreich et al.~\cite{GGLRS00}
in their seminal paper on testing monotonicity. For Boolean functions $f:\{0,1\}^n \to \{0,1\}$,~\cite{GGLRS00} described a non-adaptive $O(n^{3/2}/\eps)$-query tester
(while for monotonicity of Boolean functions they described a non-adaptive $O(n/\eps)$-query tester). To our knowledge, there was no further progress on this problem, until the recent result of Khot and Shinkar~\cite{KhSh16} that gives an adaptive $O(n\log (n)/\eps)$-query tester for Boolean functions.
Our main theorem is the following. 

\begin{theorem} \label{thm:main} Consider functions $f:\{0,1\}^n \mapsto R$, where
$R$ is an arbitrary ordered set.
There exists an one-sided error, non-adaptive, $O((n/\eps) \log(n/\eps))$-time tester
for unateness.
\end{theorem}

%\subsection{Connections to previous work} \label{sec:prev}

Monotonicity testing has been extensively studied in the past two decades~\cite{GGLRS00,DGLRRS99,EKK+00,HK03,HK04,BGJRW09,BCG+10,ChSe13-j,ChSe13,BeRaYa14,ChenST14,ChenDST15,ChDi+15,BeBl16}. 
We employ 
a previous result on testing derivative bounded properties by Chakrabarty et al.~\cite{ChDi+15}.
That paper provides general theorems about the testability of properties that are specified
in terms of the partial derivative being bounded. 

\begin{definition}
	Given an $n$-dimensional bit vector ${\bf b}$, call a function $f:\{0,1\}^n \to R$ ${\bf b}$-monotone 
	if for all $i$ with ${\bf b}_i = 0$ we have $\partial_i f(x) \geq 0$ for all $x$, and for all $i$ with ${\bf b}_i = 1$ we have $\partial_i f(x) \leq 0$ for all $x$.
\end{definition}
\noindent
Note that ${\bf 0}$-monotonicity is simply the standard notion of monotonicity.
Also note that a function is unate iff it is ${\bf b}$-monotone for some ${\bf b}$.

The property of ${\bf b}$-monotonicity is a \emph{derivative-bounded property},
in the language of~\cite{ChDi+15}.
A dimension reduction theorem for derivative properties (Theorem 8 in Arxiv version of~\cite{ChDi+15}),
when instantiated for ${\bf b}$-monotonicity, implies the following theorem.

\begin{theorem} \label{thm:dimred} Fix bit vector ${\bf b}$ and function $f:\{0,1\}^n \to R$.
	Let $\eps$ denote the  distance of $f$ to ${\bf b}$-monotonicity.
Let $\mu_i$ be the fraction of points where $\partial_i f$ violates
${\bf b}$-monotonicity, that is, the number of hypercube edges across dimension $i$ which violate ${\bf b}$-montonicity is $\mu_i 2^n$.
Then $\sum_{i=1}^n \mu_i \geq \eps/4$.
\end{theorem}

The above theorem (in fact a stronger version without the $4$ in the denominator) can also be obtained by observing that $f$ is $\bb$-monotone (resp, $\eps$-far from being $\bb$-monotone) iff the function $g(x) := f(x\oplus b)$ is monotone (resp, $\eps$-far from being monotone). Every hypercube edge that violates monotonicity for $g$ violates $\bb$-monotonicity for $f$.
A previous result of the authors shows that if a function $g$ is $\eps$-far from being monotone, then it has $\eps 2^{n-1}$ hypercube edges violating monotonicity~\cite{ChSe13}. (Such results were previously known for the case of Boolean range~\cite{GGLRS00}, and weaker results for general range~\cite{DGLRRS99}.  Refer to~\cite{ChSe13} for more details.) 

One can show that it suffices to query $\partial_i f$ at $O(1/\mu_i)$ points to detect a violation to unateness. We need to ``interpolate" between two opposite scenarios: exactly one $\mu_i = \Omega(\eps)$ and all others are $0$ versus all $\mu_i = \Theta(\eps/n)$. An efficient strategy for achieving this is Levin's investment strategy (refer to Section 8.2.4 of Goldreich's book~\cite{Go-book}). A tighter analysis of this method is given by Berman et al~\cite{BeRaYa14}, which is effectively what we use.
For the sake of completeness, we repeat the calculations of~\cite{BeRaYa14} for a complete proof.

\section{The Tester} \label{sec:tester}

\medskip
\fbox{
\begin{minipage}{0.9\textwidth}
{\tt \unatetest$(f,\eps)$}

\smallskip
\begin{compactenum}
    \item For $r = 1, 2, \ldots, L := \cei{\log(8n/\eps)}$: \\
%    \begin{compactenum}
         Repeat $s_r = \cei{\frac{20n}{\eps\cdot 2^r}}$ times
        \begin{compactenum}
            \item \label{step:samp} Sample u.a.r. dimension $i$.
            \item \label{step:set} Sample a set $R_i$ of $3 \cdot 2^r$ u.a.r. points in the hypercube and evaluate
            $\partial_i f$ at all these points.
            \item \label{step:rej} If there is some $x \in R_i$ such that $\partial_if(x) > 0$ {\bf and}
           $y \in R_i$ such that $\partial_i f(y) < 0$, reject and abort.
        \end{compactenum}
%    \end{compactenum}
    \item Accept (since tester has not rejected so far)
\end{compactenum}
\end{minipage}}

\medskip

It is evident that this is a non-adaptive, one-sided tester. Furthermore, 
the running time is $O((n/\eps)\log(n/\eps))$. It suffices to prove
the following.

\begin{theorem} \label{thm:rej} If $f$ is $\eps$-far from being unate,
\unatetest{} rejects with probability at least $1-1/e$.
\end{theorem}

\begin{proof} For dimension $i$, let $U_i$ be the set of points in $\{0,1\}^n$
where $\partial_i f(x) > 0$. Analogous, let $D_i$ be the set of points where $\partial_i f(x) < 0$.
The tester rejects iff it finds a triple $i, x, y$ such that $x \in U_i$
and $y \in D_i$.  Let $\mu_i := \min(|U_i|, |D_i|)/2^n$.

Define the $n$-dimensional bit vector ${\bf b}$ as follows: ${\bf b}_i = 0$
if $|U_i| > |D_i|$ and $\bb_i = 1$ otherwise. 
Observe that the fraction of points
where $\partial_if$ violates ${\bf b}$-monotonicity is precisely $\mu_i$.
Since $f$ is $\eps$-far from being unate,
$f$ is $\eps$-far from being ${\bf b}$-monotone. 
By \Thm{dimred}, $\sum_i \mu_i \geq \eps/4$.\smallskip

\noindent
For any integer $r \geq 1$, define $S_r := \{i \in [n] ~:~ \mu_i \in (1/2^r, 1/2^{r-1}]\}$. 

\begin{claim} \label{clm:sum} $\sum_{r=1}^L  |S_r|/2^r \geq \eps/16$.
\end{claim}

\begin{proof}
Observe that
$\sum_{r > L} |S_r|/2^r \leq (\eps/8n) \sum_r |S_r| = \eps/8$ since $\sum_r |S_r| = n$.
Since for any $i\in S_r$ we have $\frac{1}{2^{r-1}} \geq \mu_i$, we get that
$\sum_r |S_r|/2^r \geq \sum_i \mu_i/2  \geq \eps/8$ from Theorem~\ref{thm:dimred}. We subtract these bounds
to prove the claim.
\end{proof}

Fix $r$. Let $p_r$ be the probability that \Step{samp}, \Step{set}, and 
\Step{rej} reject for this $r$. Then the probability that the tester rejects is
\begin{equation}\label{eq:1}
1 - \prod_{r=1}^L (1 - p_r)^{s_r} \geq 1 - e^{-\sum_{r=1}^L p_rs_r}
\end{equation}
We now lower bound $p_r$. The tester rejects iff the set $R_i$ in~\Step{set} contains a point in $U_i$ and in $D_i$.
The probability that $R_i$ {\em does not} contain a point in $U_i$ or $D_i$ is at most 
$\left(1 - |U_i|/2^n\right)^{|R_i|} + \left(1 - |D_i|/2^n\right)^{|R_i|} \leq 2(1 - \mu_i)^{|R_i|}$.
Note that if the sampled dimension $i$ lies in $S_r$, then this probability is at most $2\left(1 - 1/2^r\right)^{3\cdot 2^r} < 1/6$.
Therefore, we get $p_r > \frac{5}{6}\cdot \frac{|S_r|}{n}$. Since $s_r \ge \frac{20n}{\eps\cdot 2^r}$, we get 
\[
\sum_{r=1}^L p_rs_r \geq \sum_{r=1}^L \frac{5}{6}\cdot\frac{|S_r|}{n} \cdot \frac{20n}{\eps 2^r}\geq \frac{100}{6\eps}\sum_{r=1}^L \frac{|S_r|}{2^r} > 1
\]
where the second inequality follows from Claim~\ref{clm:sum}.
Substituting in~\eqref{eq:1}, we get the theorem.
%If the dimension $i$ chosen in~\Step{samp} lies is in $S_r$,
%then the probability $R_i$ contains a point in $U_i$ is 
%
%then the probability of finding a point in $U_i$ (or $D_i$) is at least $\mu_i \geq 1/2^r$.
%A simple Chernoff bound and union bound show that the probability of rejection
%is at least $5/6$.
%Observe that $\EX[X_{r,s}] \geq \Pr[i \in S_r]\Pr[X_{r,s} = 1 | i \in S_r]
%\geq (|S_r|/n)\cdot (5/6)$. By linearity of expectation and \Clm{sum}, 
%$$ \EX[\sum_{r,s} X_{r,s}] \geq \sum_{r \leq \log(8n/\eps)} (1000n/\eps2^r) \cdot (|S_r|/n) \cdot (5/6)
%= (500/\eps) \sum_{r \leq \log(8n/\eps)} |S_r|/2^r \geq 50$$
%Since $X_{r,s}$s are Bernoullis, a Chernoff bound proves that $\Pr[\sum_{r,s} X_{r,s} > 1] \geq 5/6$.
%This completes the proof.
\end{proof}

\section{Acknowledgements}

We thank Oded Goldreich for pointing out the connections to Levin's investment strategy.

\bibliographystyle{alpha}
\bibliography{derivative-testing}

\end{document}